\documentclass[a4paper,10pt]{article}
\usepackage[utf8]{inputenc}
\usepackage{amsmath}
\usepackage{hyperref}
\usepackage{graphicx}
\usepackage{amsmath,amssymb,amsthm} 
\usepackage{enumerate}
\usepackage{amsthm}

\theoremstyle{plain} 
\newtheorem{theorem}{Theorem}[section]
\newtheorem{lemma}[theorem]{Lemma}
\newtheorem{proposition}[theorem]{Proposition}

\theoremstyle{definition}
\newtheorem{definition}{Definition}[subsection]

\theoremstyle{remark}

\DeclareGraphicsExtensions{.pdf,.png,.jpg}
\title{Capital allocation and risk appetite under Solvency II framework}
\author{Paolo De Angelis\\ 
Sapienza University of Rome  \\ 
Via Del Castro Laurenziano 9 Rome 00161 - Italy \\
Email: \texttt{paolo.deangelis@uniroma1.it}
\medskip\\
Ivan Granito\\
Sapienza University of Rome \\
Viale Regina Elena 295/G, 00161, Rome (Italy) \\
Email: \texttt{ivan.granito@uniroma1.it}
}

\begin{document}
\maketitle
\begin{abstract}

The aim of this paper is to introduce a method for computing the allocated Solvency II Capital Requirement (SCR) of each Risk which the company is exposed to, taking in account for the diversification effect among different risks. The method suggested is based on the Euler principle. We show that it has very suitable properties like coherence in the sense of Denault (2001) and RORAC compatibility, and practical implications for the companies that use the standard formula. 
Further, we show how this approach can be used to evaluate the underwriting and reinsurance policies and to define a measure of the Company's risk appetite, based on the capital at risk return.

\end{abstract}

\clearpage

\section*{Introduction}
The Solvency II directive requires that insurance undertakings have to calculate the \emph{Solvency Capital Requirement} taking into account for the correlation among the risk driver. This implies the existence of a diversification effect.

The evaluation of the Solvency II Capital Requirement net of diversification effect is a needful procedure to know the real capital absorption of the lines of business and to evaluate the relative financial performance. 

Academic research addresses the capital allocation for many years. Indeed they were formulated various approaches to the problem by moving from game theory or establishing the principles of coherence through axiomatic definitions for evaluating allocation methods in relation to the specific risk measures. This last line of research has provided significant applications respect to various risk measures assuming different distributions for the underlying risk variable and identifying the Euler's allocation principle as the highest performing. \\
The most important papers we refer are:
\begin{itemize}
\item Tasche (1999) \cite{Tas1} define the RORAC compatibility as the most important economic property of an allocation principle and state that, for risk measures with continuous derivatives, the unique continuous per-unit allocation principle RORAC compatible is that of Euler
\item Denault (2001) \cite{Den1} establishes the principles of coherence for an allocation principle and derive the Euler allocation principle moving from game theory
\item A. Buch, G. Dorfleitner (2008) \cite{BucDor} state that the Euler allocation principle associated with a coherent risk measure produce a coherent allocation of risk capital
\end{itemize}
The aim of this paper is to study the Solvency II capital requirement allocation for European insurance companies that calculates the SCR by means of the Standard Formula providing an allocation principle and an approach to evaluate the financial performance of the risk capital invested.
\\ \\
Our way it is to consider the SCR as risk measure noticing that, under the set of hypothesis underlying the standard formula, is coherent in the sense of Artzner (1999) \cite{Art2}. Then, by means of the Euler's allocation principle, we derive the closed formulas to calculate the allocated SCR among the risk considered in the multilevel aggregation scheme of Solvency II standard formula. Due to the cited results we know that the allocation provided is coherent in the sense of Denault (2001) \cite{Den1} \footnote{see A. Buch et G. Dorfleitner (2008) \cite{BucDor}} and RORAC compatible \footnote{see Tasche (1999) \cite{Tas1}}. Then we show that, given the RORAC compatibily, this result can be used to evaluate the financial performance of an insurance portfolio.
\\ \\
The paper is organized as follows. In the section \ref{Par:Theory} we introduce all theoretical background used in the following sections as Euler Theorem, coherence of risk measures (\cite{Art2}), coherence of risk capital allocation (\cite{Den1}), RORAC compatibility and coherence of the Euler's allocation (\cite{Tas1} and \cite{BucDor}). In the section \ref{SEC:hp FS} we describe the Standard Formula approach to SCR calculation and show a set of coherent hypothesis and definitions. In the section \ref{sec:CORPO} we define the diversification effect as variable and we provide the formulas for the SCR allocation among each single macro and micro risk included in the multilevel aggregation scheme of the standard formula. In the section \ref{sec:RISK APPETITE} we provide a mean variance model for the RORAC to evaluate the underwriting and reinsurance policies and to define the risk appetite on each sub-portfolios by solving an optimization problem. Follows the conclusion and the perspectives for future research.
\section{Theoretical framework} \label{Par:Theory}
We consider an insurance company whose portfolio of insurance contract is composed by $q$-homogeneous sub-portfolios. We define a set of random variable $\Gamma$ in the probability space $[\Omega, \Im ,\emph{\textbf{P}}]$.
The risk of the sub-portfolio \emph{s-th} ($s=1 ... q)$ is modelled by means of the generic random variable $X_s \in \Gamma$ . The total risk of the company is described with the random variable $X=\sum\limits_{s=1}^q X_s$. The company calculates its regulatory capital requirement by means of a risk measure defined as $\pi(X) : \Gamma \rightarrow \Re$. 

Note that the risk variables $X_s$ are dependent so there exist a diversification effect implied in the calculation of the capital requirement $\pi(X)$.

\subsection{Coherence of risk measure}

Artzner (1999) \cite{Art1} introduced the definition of coherent risk measure by means of the following axiom:

\begin{definition}
A risk measure $\pi$ is considered coherent if satisfies the following property:
\begin{itemize}
\item \textbf{Traslation invariance}: for a riskless deterministic portfolio $L$ with fixed return $\alpha$ and for all $X \in \Gamma$ we have $\pi(X+L)=\pi(X)-\alpha$
\item \textbf{Subadditivity}: for all $(X_1,X_2) \in \Gamma$ we have $\pi(X_1 + X_2) \leq \pi(X_1) + \pi(X_2)$
\item \textbf{Positive Homogenity}: for all $\lambda >0$ and all $X \in \Gamma$, $\pi(\lambda X) = \lambda \pi(X)$
\item \textbf{Monotonicity}: for all $X,Y \in \Gamma$ with $X \leq Y$, we have $\pi(X) \leq \pi(Y)$
\end{itemize}
\end{definition}

\subsection{Coherence of allocation principle}

Denault (2001) \cite{Den1} extends the concept of coherence to the allocation principle establishing a set of definitions and axioms. 

We consider a set of \emph{q} portfolios. The relative allocated risk measures to be calculated represent a set $A$ of allocatoin problem. The following definition holds:

\begin{definition} \label{DEF:AllocationPrinciple}
An allocation principle is a function $\Pi: A \rightarrow \Re ^q$ that maps each allocation problem into a unique allocation:

\begin{equation}
\Pi(A)=\Pi 
\left(
\begin{bmatrix}
    \pi(X_1) \\       
    \vdots  \\
    \pi(X_q)
\end{bmatrix} 
\right)
=
\begin{bmatrix}
    \pi(X_1|X) \\       
    \vdots  \\
    \pi(X_q|X)
\end{bmatrix} 
\end{equation}
\\
\\
such that $\pi(X)=\sum\limits_{s=1}^q \pi(X_s|X)$ where $\pi(X_s|X)$ is the allocated risk measure for the sub-portfolio $s-th$.
\end{definition}

\begin{definition}
An allocation principle $\Pi$ is coherent if, for every allocation problem satisfies the followings three properties:

\begin{enumerate}
\item \textbf{No Undercut}

\[
\forall M \subseteq Q, \qquad \sum_{s \in M} \pi(X_s) \leq \pi(\sum_{s \in M} X_s)
\]

\item \textbf{Symmetry}: if by joining any subset $M \subseteq Q \ {i,j}$, portfolios \emph{i} and \emph{j} both make the same contribution to the risk capital, then $K_i = K_j$.

\item \textbf{Riskless allocation}: for a riskless deterministic portfolio $L$ with fixed return $\alpha$ we have that 
\[
\pi(L) = -\alpha
\]
\end{enumerate}
\end{definition}

\subsection{Euler allocation principle and \emph{RORAC} compatibility}

As in definition \ref{DEF:AllocationPrinciple} $\pi(X)=\sum\limits_{s=1}^q \pi(X_s|X)$
where, from an economic point of view, $\pi(X_s | X)$ $(s=1, ... , q)$ is the risk contribution net of diversification effect of the $q$-sub-portfolios. The knowledge of the risk contribution $\pi(X_s|X)$ enables to evaluate the risk return profile of each sub-portfolio by defining the random variables $RORAC_s$ (Return on Risk Adjusted Capital).
\begin{definition}
\label{RORAC}
Let $U$ and $U_s$ $(s=1,...,q)$ respectively the r.v. one-year wide portfolio income and \emph{s-th} sub-portfolio income so that $U=\sum\limits_{s=1}^q U_s$, we can define:
\begin{itemize}
\item the \emph{Return on Risk Adjusted Capital} of the wide portfolio of the company as:
\begin{equation}
RORAC(X)=\frac{U}{\pi(X)}
\end{equation}

\item the \emph{Return on Risk Adjusted Capital} of the sub-portfolio \emph{s-th} of the company as:
\begin{equation}
RORAC(X_s)=\frac{U_s}{\pi(X_s|X)}
\end{equation}
\end{itemize}
\end{definition}
It is important to note that the denominator of the RORAC must represent a capital net of diversification, if is not, the RORAC have not any economic sense. This depends on a particular property of the allocation methodology used to calculate the variables $\pi(X_s|X)$ named \emph{RORAC compatibility} introduced by Tasche (1999) \cite{Tas1}.
\begin{definition}
\label{DEF:RORAC_compatibility}
A risk capital contribution $\pi(X_s|X)$ to the overall risk contribution $\pi(X)$ is RORAC compatible if there are some $\epsilon>0$ such that:

\begin{equation}
RORAC(X_s)>RORAC(X) \Rightarrow RORAC(X + hX_s)>RORAC(X)
\end{equation}

for all $0<h<\epsilon$.
\end{definition}
Tasche (1999 and 2004) find that (\cite{Tas1} and \cite{Tas2}), if a RORAC compatibility capital allocation exists, it is uniquely determined by the Euler's principle\footnote{Note that we provide a version of the Euler's theorem slightly different to that of Tasche because we refer to a risk measure, the SCR, that is a differentiable function}.
\begin{lemma}[Euler's principle]
\label{LEMMA:Tasche_Euler_Theorem}
Let $\pi(X)$ be a risk measure and assume that it is a 1-degree homogeneous and continuously differentiable function. If there are risk contributions $\left[\pi(X_1|X), ... , \pi(X_q|X) \right]$ that are RORAC compatible, they are uniquely determined as:

\begin{equation}
\pi_{Euler}(X_i|X) = \pi(X_i) \cdot \frac{\partial \pi(X)}{\partial \pi(X_i)} \qquad i=1, ... , q
\end{equation}

This is called Euler's allocation principle of the risk measure $\pi(X)$ among the $q$-sub-portfolios. 
\end{lemma}
From a mathematical point of view, the Euler allocation principle derives from the application to the risk measures considered of the well known Euler's Homogeneous Function Theorem.
\subsection{On the coherence of Euler allocation principle}

The Euler's allocation principle described in the previous subsection, is one of the most popular allocation method proposed in the literature. This is due for its suitable properties. In this way, a very important contribution is that of Buch et G. Dorfleitner (2008) \cite{BucDor}. From an axiomatic point of view, they study the relation between the properties of the Euler's allocation principle and those of the risk measure to which the allocation is applied. What they find is resumed in the following proposition.
\begin{proposition}
\label{PROP:Buch_Dorfleitner}

The Euler's allocation principle applied to a coherent risk measure has the properties of full allocation, “no undercut” and riskless allocation" so it is coherent in the sense of Denault (2001) \cite{Den1}.
\end{proposition}

This result has a main role for our research because we will use it to prove that the allocation methodology that we find to calculate the allocated SCR in the Solvency II Standard Formula framework by means of the Euler's principle, is coherent in the sense of Denault (2001) \cite{Den1}.

This, united to the RORAC compatibility ensured by the Euler's principle\footnote{see Tasche (2008) \cite{Tas2}}, and the closed formulas implies very suitable properties and practical implications for our results.

\section{Solvency II standard formula: hypothesis and framework} \label{SEC:hp FS}

The Solvency II directive\footnote{Directive 2009/138/EC of the European parliament and of the council 25 November 2009 on the taking-up and pursuit of the business of Insurance and Reinsurance (Solvency II)} provide that the insurance companies have to calculate their regulatory capital requirement, named Solvency Capital Requirement (hereafter referred as SCR), by means of a risk based methodology. From a practical point of view, they can choose between a standard formula provided by EIOPA or to produce him self an internal model. In the following we take into account only for the case that the SCR is calculated with the standard formula (hereafter referred as FS). The FS provides that the Company should calculate the SCR through the modular approach that will be defined.
\\
\\
The risk-based modular approach considered in the Solvency II framework provides that the company has to consider the global risk which is exposed to by dividing it into single components. These components are related with different sources of risk (hereafter named "risk") like reserve risk, mortality risk, interest rate risk e so on. 
The modular scheme considers $n$ macro risks. The generic macro risk $i-th$ ($i=1, ... , n$) is composed by $m_i$ micro risk. We use the following notation for all variables will be defined: the first digit of the subscript identifies the macro-risk and is from $1$ to $n$, the second one identifies the micro-risk and is from $1$ to $m_i$ (where $i$ identify the overlying macro-risk). \\ \\
We define a set of random variable $\Psi$ on the probability space $[\Omega, \Im, \emph{\textbf{P}}]$. Let $L_{ij} \in \Psi$ (with $ij=i1, ... ,im_i$) the random variable that describe losses, over an annual time horizon, associated  with the micro-risk $ij-th$ and let $Y_{ij} = L_{ij} - E(L_{ij})$ the respective r.v. unexpected losses. The generic macro risk $i-th$ depends of a random variable $Y_i=\sum\limits_{j}^{m_i} Y_{ij}$. 
The total risk of the company\footnote{We not consider some parts of the modular scheme defined in the FS like the adjustment for deferred taxes and adjustment for loss absorbing capacity of technical provision. This because of their effect is measurable after calculating the SCR and their allocation depends not of the aggregation scheme but of particular consideration made by the Company.} is described with the random variable $Y=\sum\limits_{i=1}^n Y_i=\sum\limits_{i=1}^n\sum\limits_{j=1}^{m_i} Y_{ij}$. 

\begin{definition}[Standard Formula] The solvency II capital requirement of the company is defined by means of the specific risk measure on $Y$ following defined: 
\label{def:SF}
\begin{enumerate}[I]
\item $SCR_{ij}$ is the capital requirement referred to the $ij-th$ micro-risk defined as:
 \begin{equation}
SCR_{ij} = VaR_{99.5\%}(Y_{ij})
\end{equation}

is approximated by means of specific formulas provided by EIOPA.

 \item $SCR_{i}$ is the capital requirement referred to the $i-th$ macro-risk calculated by aggregating the underlying micro-risk:
 
\begin{equation} \label{EQ:SCR_ix}
SCR_{i}=\sqrt{\sum_{x=1}
^{m_i} \sum_{y=1}
^{m_i} SCR_{ix} \cdot SCR_{iy} \cdot \rho_{ix,iy}} 
\end{equation}

where $\rho_{ix,iy}$ represents the linear correlation coefficients provided by EIOPA.

 \item $SCR$ is the overall capital requirement of the company and is calculated by aggregating the underlying macro-risk:
 
\begin{equation} \label{EQ:SCR_i}
SCR= \sqrt{\sum_{i=1}
^{n} \sum_{w=1}
^{n} SCR_{i} \cdot SCR_{w} \cdot \rho_{i,w}} 
\end{equation}
\\
where $\rho_{i,w}$ represents the linear correlation coefficients provided by EIOPA.
\end{enumerate}
\end{definition}

Note that the square-root aggregation formula (\ref{EQ:SCR_ix} and \ref{EQ:SCR_i}) implies that the r.v. $Y_{ij}$ ($j=1,...,m_i$) are jointly normal distributed and linearly correlated \footnote{For the sake of clarity, some further specification are needed. In the FS there is not any explicit hypothesis for the distribution of the losses for the micro and macro risk, but the underlying assumption of linear correlation and normal distribution is needed. These assumptions are very strong because, in the insurance problems, the linear correlation and the non skewed distribution hypothesis are not always consistent with the nature of the phenomena considered. Indeed, Sandstrom (2007) \cite{Sandstrom1} specified that, for skewed distribution, the normal approximation can implies an incorrect estimation of the SCR and propose a method to transform, by Cornish-Fisher expansion, the quantile distribution from a skewed into a standard normal distribution. Instead, from a practical point of view, the standard formula of Solvency 2, avoid this problem through a very strong and prudential simplification that is by overestimating the confidence level for the calculation of the value at risk for some micro-risk variable. In particular, for a normal distribution holds that:
\[
VaR_{99.5\%}(Y)=\alpha \cdot \sigma_X
\]

where $\alpha=F_X^{-1}(0.995)=2.576$ with $F$ the cumulative distribution function of standard normal random variable. In the standard formula, e.g. for the premium-reserve risk for non-life portfolios, $\alpha = 3$ that corresponds to a  $VaR_{99.87\%}(Y)$ avoiding in this manner the skewness under estimation due to normally assumption. Furthermore EIOPA itself has specified  that the correlation parameters provided for the FS, are estimated to reduce the distorting effect through the following formulation:
\[
\min_{\rho} VaR(X+Y)^2 - VaR(X)^2 - VaR(X)^2 -2 \rho VaR(X) VaR(Y)
\]} so that the SCR is a coherent risk measure.
\\
\begin{figure}[htbp]
\centerline{\includegraphics[width=14cm]{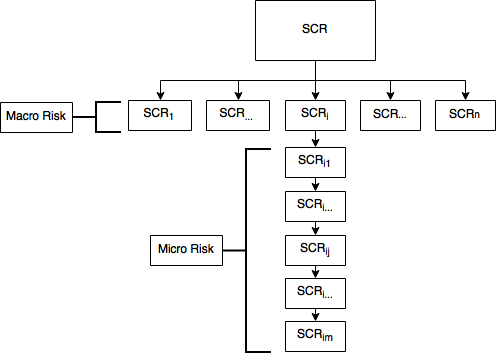}}
\caption{Aggregation scheme}
\end{figure} 
\\
\clearpage
\section{Capital allocation} \label{sec:CORPO}

With reference to the SCRs of the macro-risk and micro-risk, moving from Lemma \ref{LEMMA:Tasche_Euler_Theorem} the allocation formulas are obtained.

\begin{theorem}[SCR macro-risk allocation] \label{MacroRiskTheorem}
In the case of Solvency II Standard Formula, the RORAC compatible allocation of the overall SCR among the constituents macro-risk is uniquely determine as:

\begin{equation}\label{eq:SCR_Eul1} 
  SCR(Y_i|Y)= SCR_i \cdot \frac{\sum\limits_{w=1}^n  SCR_w \cdot \rho_{i,w}}{SCR_Y}
\end{equation}

where $SCR(Y_i|Y)$ is the allocated \emph{i-th} macro-risk.

\end{theorem}
\begin{proof}
From Lemma \ref{LEMMA:Tasche_Euler_Theorem} (Euler's allocation principle) holds that:
\begin{equation}
SCR_{Euler}(Y_i|Y) = SCR_i \cdot \frac{\partial SCR_Y}{\partial SCR_i} 
\end{equation}
\\
\\
where the partial derivative is:

 \begin{equation}
\frac{\partial SCR_Y}{\partial SCR_i}= \frac{\sum\limits_{w=1}^n  SCR_w \cdot \rho_{i,w}}{SCR_Y} 
\end{equation}
 
and so:

\begin{equation}
  SCR(Y_i|Y)= SCR_i \cdot \frac{\sum\limits_{w=1}^n  SCR_w \cdot \rho_{i,w}}{SCR_Y}
\end{equation}
\end{proof}
Moving from theorem \ref{MacroRiskTheorem} it is possible to reach a similar result for the micro-risk allocation. It is useful to define the variable {Allocation Ratio} as: 
\begin{equation}
AR_i = \frac{SCR(Y_i|Y)}{SCR_i}
\end{equation}
\\
\\
\begin{theorem}[SCR micro-risk allocation] \label{MicroRiskTheorem}
In the case of Solvency II Standard Formula, the RORAC compatible allocation of the macro-risk SCR $i-th$ $(i=1,...,n$) among the micro-risk $ix-th$ $(x=1,...,m_i)$ is uniquely determined as:

\begin{equation}\label{eq:SCR_Eul2} 
  SCR(Y_{ix}|Y,Y_i)= SCR_{ix} \cdot \frac{\sum\limits_{y=1}^{m_i}  SCR_{iy} \cdot \rho_{ix,iy}}{SCR_i} \cdot AR_i
\end{equation}
where the variable $SCR(Y_{ix}|Y,Y_i)$ is the allocated $ix-th$ $(x=1,...,m_i)$ micro-risk. 
\end{theorem}
\begin{proof}
From Lemma \ref{LEMMA:Tasche_Euler_Theorem} (Euler's allocation principle) we have that: 
\begin{equation}
SCR(Y_{ix}|Y,Y_i)=SCR_{ix} \cdot \frac{\partial SCR_Y}{\partial SCR_{ix}}
\end{equation}

By means elementary algebra holds that:

 \begin{equation} \begin{split}
\frac{\partial SCR_Y}{\partial SCR_{ix}} &= \\ \\ &= \frac{\partial SCR_Y}{\partial SCR_{i}} \cdot \frac{\partial SCR_i}{\partial SCR_{ix}} = \\ \\ &= \frac{\sum\limits_{w=1}^n  SCR_w \cdot \rho_{i,w}}{SCR_Y} \cdot \frac{\sum\limits_{y=1}^{m_i}  SCR_{iy} \cdot \rho_{ix,iy}}{SCR_i} \cdot \frac{SCR_i}{SCR_i} \\ \\ &= \frac{\sum\limits_{y=1}^{m_i}  SCR_{iy} \cdot \rho_{ix,iy}}{SCR_i} \cdot AR_i
\end{split}\end{equation}

so that:

\begin{equation}
  SCR(Y_{ix}|Y,Y_i)= SCR_{ix} \cdot \frac{\sum\limits_{y=1}^{m_i}  SCR_{iy} \cdot \rho_{ix,iy}}{SCR_i} \cdot AR_i
\end{equation}
\end{proof}
The theorems \ref{MacroRiskTheorem} and \ref{MicroRiskTheorem} enables to conclude that, under assumptions (\ref{def:SF}), the RORAC compatible and coherent allocation of SCR is uniquely determined by means the Euler's principle and can be expressed by means of the closed expressions reported.

\section{Procedure for underwriting policy evaluation and risk appetite measurement} \label{sec:RISK APPETITE}
The risk based approach for the Solvency II capital requirement calculation enables insurance companies to evaluate their profitability taking into account for the capital absorption of the each sub-portfolio. Furthermore, the Solvency II directive requires that insurance undertaking have to evaluate their underwriting and reinsurance policies and to define the limits for the risk appetite.
In order to do it, we propose to lead back the problem to the classical portfolio theory. In this way we show that it is possible to use the same integrate framework for all the named evaluation. In particular, we consider a mean-variance model on the sub-portfolios RORAC\footnote{In the previous section RORAC compatible SCR allocation among risk are provided. To evaluate the sub-portfolios' RORAC, it is necessary to allocate the micro-risk among sub-portfolios. We not consider this problem in the paper but it can be done starting from the specific methodologies provided in the standard formula for the calculation of the several $SCR_{ix}$.}. 
\\

To evaluate different underwriting and reinsurance strategies according with the defined risk appetite, we propose the following optimization problem:

\begin{equation*}
\begin{aligned}
& \underset{P,R}{\text{maximize}}
& \mathrm{E}(RORAC) \\ \\
& \text{subject to}
&  \beta<SCR<\gamma \\
&& \delta<P<\epsilon \\
&& CV<\bar{\alpha} \\
&& const. on R
\end{aligned}
\end{equation*}

where:

\begin{description}

\item[-] $E(RORAC)$ is the expected global \emph{RORAC} of the company

\item[-] CV is the vector with coefficient of variation of the RORAC:

\begin{center}
\begin{tabular}{c}
$CV = 
\begin{bmatrix}
    \frac{\sigma(RORAC_1)}{E(RORAC_1)} \\   
    \vdots  \\
    \frac{\sigma(RORAC_n)}{E(RORAC_n)}
\end{bmatrix}
$
\end{tabular}
\end{center} 

\item[-] P is the vector of LoBs premium. $\delta$ and $\epsilon$ are the contraints for the premium depending on the commercial strategy of the Company 
\item[-] R is reisurance program subjected to qualitative and quantitative contraints
\item[-] $\alpha$ is risk appetite limit 

\item[-] $\beta$ and $\gamma$ are the bounds imposed to the overall SCR to limit the solution to wich compatible with the capital availabilities of the Company
\end{description}
From a practical point of view it is sufficient to test the realistic scenario of underwriting (given the commercial power) and reinsurance (given the market offer) and choose the global strategy that are optimal. 
In the following table, we show a numerical example based on a risk profile measurement RORAC compatible derived from an anonymous non-life company data base:
\\
\begin{center}
\begin{tabular}{c c c c}
LoB & $E(RORAC_i)$ & $\sigma (RORAC_i)$ & $SCR^A_i$  \\ 	\hline
Medical Expenses & 14,69 \% & 2,6 \% & 1.412 \\
Income Protection & -0,91 \% & 6,5 \% & 3.550 \\  
Motor vehicle liability & 9,2 \% & 11,6 \% & 4.187 \\  
Other motor & 5,94 \% & 1,8 \% & 2.977 \\  
Marine, aviation and transport & 3,95 \% & 1,4 \% & 2.115 \\  
Fire and other damage to property & 14,17 \% & 5,3 \% & 1.577 \\  
General liability & 14,6 \% & 4,4 \% & 3.366 \\
Credit and suretyship & 1,00 \% & 0,70 \% & 816 \\  
Legal expenses & 13,57 \% & 6,4 \% & 2.647 \\  
Assistance & 11,22 \% & 10,3 \% & 1.070 \\  
Miscellaneous financial loss & 14,61 \% & 5,6 \% & 4.573 \\  
Total & 9,5 \% & 5,7 \% & 28.294 \\
\end{tabular}
\end{center} 
The data above can be reported in the figure \ref{RiskReturn} that represents the contribution of each LoB risk performance to the company's risk situation.
\begin{figure}[htbp]
\centerline{\includegraphics[width=15cm]{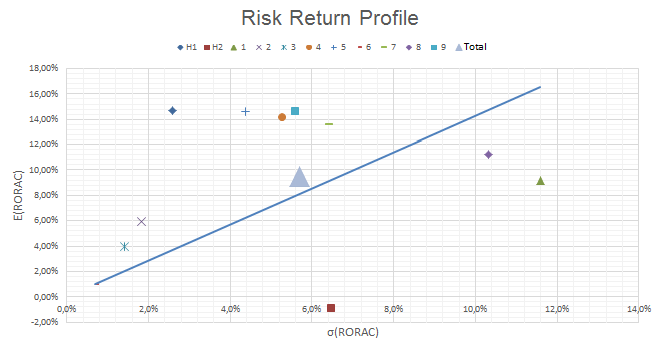}}
\caption{Risk-return profile}\label{RiskReturn}
\end{figure}
\clearpage
\
\

\section{Conclusion}
In this paper we have shown that, under solvency II standard formula framework, is possible to obtain a Solvency Capital Requirement allocation among micro and macro risks that is coherent in the sense of Denault \cite{Den1} and RORAC compatible \cite{Tas1}. We demonstrated the results by means of the Euler's allocation principle.
\\
Then, given the allocated SCR, we have provided a procedure to evaluate the underwriting and reinsurance policies and to determine the risk appetite of the stakeholder by means of a RORAC index and collocating the argument under the classical portfolio theory. 
\\
Some possible developments can be the construction of a model for the micro risk allocation among sub-portfolios for all the case in which is not possible to apply the formulas we provided (eg. the allocation of the interest rate risk among life LoBs etc.). This may have a very strong effect in the real analysis due for its straightforward possibility of application eg. the valuation of underwriting strategy and reinsurance strategy.



\end{document}